\pdfoutput=1
\RequirePackage{ifpdf}
\ifpdf 
\documentclass[pdftex]{sigma}
\else
\documentclass{sigma}
\fi

\usepackage{enumerate}
\usepackage{epigraph}
\usepackage{dsfont}

\newcommand{\Id}{\mathds{1}}

\DeclareMathOperator{\Hom}{Hom}

\newcommand{\catname}[1]{\textnormal{\texttt{#1}}}

\def\*{{\textstyle *}}
\newcommand{\R}{{\mathbb R}}

\newcommand{\C}{{\mathbb C}}
\newcommand{\Z}{{\mathbb Z}}
\def\cO{\mathcal{O}}

\numberwithin{equation}{section}

\newtheorem{Theorem}{Theorem}[section]
\newtheorem*{Theorem*}{Theorem}

\newtheorem{Lemma}[Theorem]{Lemma}
\newtheorem{Proposition}[Theorem]{Proposition}

\theoremstyle{definition}
\newtheorem{Definition}[Theorem]{Definition}

\newtheorem{Example}[Theorem]{Example}
\newtheorem{Remark}[Theorem]{Remark}

\begin{document}

\allowdisplaybreaks

\newcommand{\arXivNumber}{2604.22381}

\renewcommand{\PaperNumber}{117}

\FirstPageHeading

\ShortArticleName{Affine Supertrusses and Superbraces}

\ArticleName{Affine Supertrusses and Superbraces}

\Author{Andrew James BRUCE}

\AuthorNameForHeading{A.J.~Bruce }

\Address{Cardiff, UK}
\Email{\mail{andrewjamesbruce@googlemail.com}}

\ArticleDates{Received June 11, 2026, in final form November 04, 2026; Published online November 17, 2026}

\Abstract{Brzezi\'nski's trusses are ``ring-like'' algebraic structures in which the addition is replaced with an abelian heap operation and the binary product satisfies a natural distributivity rule over the ternary addition. This paper addresses the question of how to define affine super-versions of trusses. Taking our cue from the theory of algebraic supergroups, we define an affine supertruss as a representable functor from the category of unital associative supercommutative superalgebras to the category of trusses. The representing superalgebras are equipped with a ``cotruss'' structure--a new concept in itself. We show that from an affine supertruss, provided the binary product defines a supergroup, one can construct an affine two-sided superbrace, and so generalise Rump's braces to supermathematics. As an application of these constructions, we propose a generalisation of the set-theoretic Yang--Baxter equation to the setting of affine superschemes via a super-version of Adler's map.}


\Classification{20N10; 14M30; 16T15; 17A70}

\epigraph{Mathematics is the science which uses easy words for hard ideas.}{\textit{Edward Kasner}}

\section{Introduction and background}
\subsection{Introduction}
Recently, Brzezi\'nski and collaborators have introduced and studied trusses (see, for example, \cite{Breaz:2024,Breaz:2025,Brzezinski:2018,Brzezinski:2019,Brzezinski:2020,Brzezinski:2023,Brzezinski:2019b}). The motivation for the study of trusses is to find a unified framework to study Rump's braces \cite{Rump:2007} and rings on equal footing. We view a truss as a (two-sided) brace in which the identity element has been disregarded, or as a ring in which the zero has been removed. Thus, trusses can be seen as affine versions of braces and rings, i.e., the privileged elements are missing or not fixed. To highlight some of the key advances, trusses were first formally introduced by Brzezi\'nski \cite{Brzezinski:2018}, and their structure was further studied in \cite{Brzezinski:2020}, where the concepts of paragons (quotients), ideals and modules were defined. The representation theory of trusses was investigated in \cite{Brzezinski:2019b}; specifically, the behaviour of direct sums and free modules over trusses were compared to those over classical rings. Further aspects of the theory of heaps of modules were studied by Breaz et al.\ \cite{Breaz:2024,Breaz:2025}. Brzezi\'nski and Papworth \cite{Brzezinski:2023} generalised Hochschild cohomology to trusses. Andruszkiewicz et al. \cite{Andruszkiewicz:2022} established the direct relationship between ideal extensions of rings and trusses; this opened up the study of ideal extensions using the framework of trusses. The complete classification of trusses whose retracts are finite cyclic groups was completed by Andruszkiewicz and Pryszczepko~\cite{Andruszkiewicz:2025}.

Rump introduced braces as a tool in the study of the set-theoretic solutions of the (quantum) Yang--Baxter equation as first proposed by Drinfel'd \cite{Drinfeld:1992}. For a review of the set-theoretic Yang--Baxter equation, the reader may consult Doikou \cite{Doikou:2024}. Given the importance of superrings in supergeometry, mathematical physics, algebraic topology, and homological algebra, it is natural to wonder if one can construct a super-version of a truss and if that leads to novel generalisations of Yang--Baxter maps. Heuristically, `superising' an algebraic structure is `liberally peppering' the expressions with plus and minus signs.

A truss is a `ring-like' algebraic structure in which the addition operation is replaced by a~ternary operation--technically an abelian heap operation--and the binary product satisfies the left and right distributivity over the ternary `addition'. We stress that the ternary operation of an abelian heap should be thought of as an addition in which the zero element is absent. Generally, many of the notions and constructions from the theory of rings generalise to trusses, some more directly than others; we direct the reader to the literature already cited. As we do not have a~vector space structure underlying a truss, we cannot directly define a `supertruss' by including Koszul sign factors. In particular, a~superring decomposes as the direct sum $R = R_\mathbf{0} \oplus R_\mathbf{1}$ of even and odd elements, and from there the appropriate sign factors in an expression can be deduced. However, such a decomposition makes no sense in the setting of trusses, as the addition is replaced with an abelian heap operation. The solution here is to adopt a categorical approach following the theory of algebraic supergroups, where the structures are understood as representable functors from the category of unital associative supercommutative superalgebras to the category of groups. Thus, we define \emph{affine supertrusses} as representable functors from the category of unital associative supercommutative superalgebras to the category of trusses.

Naturally, the representing superalgebra of an affine supertruss comes with the structure of what we call a \emph{supercotruss}.\footnote{We prefer the term `supercotruss' over `quantum supertruss' or similar, noting that both the binary multiplication and ternary addition are formally dualised, and we wish to avoid any connotations of deformation quantisation of affine superschemes.} That is, the superalgebra representing an affine supertruss comes with a pair of maps encoding the dual notion of binary multiplication and ternary addition. The philosophy is that a cotruss is to a truss what a Hopf algebra is to a group. A little more carefully, a supercotruss is a superalgebra that has the structure of a coalgebra and an abelian quantum heap, together with a suitable compatibility condition. \v{S}koda introduced quantum heaps \cite{Skoda:2007}, and we comment that a related notion of a ternary Hopf algebra was introduced by Borowiec, Dudek, and Duplij \cite{Borowiec:2001}. The principal distinction between quantum heaps and ternary Hopf algebras is the precise form of the ternary coassociativity.

By adopting this categorical/functorial approach, we have left the pure set-theoretic world and entered the world of affine superschemes (see, for example, Carmeli et al.~\cite[Chapter~10]{Carmeli:2011}). Thus, geometrically, supercotrusses are understood as the coordinate rings on a superspace that, via Grothendieck's functor of points, has the structure of a truss. However, we will not push this geometric understanding in this paper and generally focus on the algebraic/categorical aspects of the theory.

We show that Brzezi\'nski's construction of a two-sided brace from a braceable truss. i.e., a~unital truss such that the binary multiplication is a group operation, generalises to the categorical framework presented in this paper (see \cite[Corollary~3.10]{Brzezinski:2020} for the classical construction). That is, via braceable affine supertrusses, we are able to construct super-versions of Rump's braces, which we refer to as \emph{affine two-sided superbraces}. We remark that for unital trusses, one obtains a brace-type distribution law, but to obtain the strict definition of a two-sided brace, the group property of the multiplication is required. Similarly to trusses, there is no underlying vector space structure associated with a (two-sided) brace, and thus, directly `superising' braces is unclear.

We propose that affine two-sided superbraces can be used to construct (families) of solutions to the set-theoretic Yang--Baxter equation following Rump \cite{Rump:2007}. However, it must be stressed that the Yang--Baxter equation in question is not `super', i.e., it is not an equation for a graded linear operator acting on a tensor product of graded modules, and thus involves no Koszul sign rule in its braiding. The `superness' is internal, and our constructions generalise the set-theoretic Yang--Baxter equation and Yang--Baxter maps to the world of affine superschemes. We suggest that these notions may find applications in discrete integrable systems with fermionic/anticommuting degrees of freedom. In particular, Grassmann extensions of Yang--Baxter maps originate in the work of Grahovski, Konstantinou-Rizos and Mikhailov (see~\cite{Grahovski:2016}) using Lax representations. The simplest starting place here are sets of the form
\[ V_\Lambda := \{(\mathsf{x}, \vartheta) \mid \mathsf{x}\in \Lambda_\mathbf{0},\, \vartheta \in \Lambda_\mathbf{1} \},\]
where $\Lambda = \Lambda_\mathbf{0} \oplus \Lambda_\mathbf{1}$ is a fixed finite-dimensional Grassmann algebra over a field $\mathbb{F}$ (typically~$\R$ or~$\C$). The physical interpretation is that the dynamical variables in a discrete system are valued in a Grassmann algebra. Grassmann extended Yang--Baxter maps are then standard set-theoretical Yang--Baxter maps $ S\colon V_\Lambda \times V_\Lambda \rightarrow V_\Lambda \times V_\Lambda$ that preserve the $\Z_2$-grading. We recognise that $V_\Lambda \cong \mathbb{F}^{1|1}(\Lambda) \cong \Hom_{\catname{SAlg}_{\mathbb{F}}}(\mathbb{F}[x, \theta], \Lambda)$, where $\mathbb{F}[x, \theta]$ is the polynomial algebra in even and odd formal variable $x$ and $\theta$, where $\theta^2 =0$. Thus, the Grassmann extended sets are best understood, in the author's opinion, in terms of Schwarz and Voronov's $\Lambda$-points (see \cite{Schwarz:1984,Voronov:1984}) where the underlying superspace is $\mathbb{F}^{1|1}$, i.e., the superspace with one even and one odd coordinate. Our suggestion of using affine two-sided superbraces to construct Yang--Baxter maps complements the construction of Grassmann extended Yang--Baxter maps. In our constructions, Grassmann algebras are replaced with general superalgebras, and as we are working functorially, the superalgebras are not fixed. This sits comfortably with the philosophy that a choice of Grassmann algebra to parameterise a physical system has no physical outcomes. However, in this paper, we will not attempt to make explicit links between the approach suggested here and the established examples of Grassmann extensions of Yang--Baxter maps. The warning here is that not all Yang--Baxter maps come from braces and so it may not be possible to reformulate all Grassmann extensions of Yang--Baxter maps in terms of affine supertrusses.

We establish that abelian supergroups and (finite projective) superrings provide examples of affine supertrusses, generalising the classic examples to supermathematics. To illustrate the formalism, we present a super-version of Adler's map; this example shows that one can work informally or symbolically with ``algebra-valued coordinates/generators'' as is common for physicists working with supergeometry.

\textbf{Key results:}
\begin{enumerate}[(i)]\itemsep=0pt
\item Theorem~\ref{trm:RepObs} establish the properties of the representing objects of affine supertrusses, which we call supercotrusses.
\item Theorem~\ref{trm:EqCat} establishes an anti-equivalence between the categories of affine supertrusses and supercotrusses.
\item Proposition \ref{prop:BraceNat} generalises Brzezi\'nski's construction of a two-sided brace from a braceable truss to the setting of affine supertrusses.
\end{enumerate}

\subsection{Further motivation} There has been a renewed interest in ternary structures in algebra and geometry. Generically, a heap can be viewed as a group in which the identity element has been removed or forgotten. Heuristically, we view a heap as an affine generalisation of a group. By picking an element of the heap, we can construct a group in which that element acts as the identity; this is known as the retraction of the heap. All such groups are isomorphic. However, we do not have an equivalence of categories, as a point needs to be selected. We view selecting an identity as analogous to choosing an observer or picking a gauge in physics. If one wants to develop physics without reference to observers, heaps and related ternary structures offer a mathematical route. For instance, Ostermann and Schmidt~\cite{Ostermann:1966} demonstrated that affine spaces can be defined without reference to vector spaces by using two ternary operations; one entirely on the set and the other representing an action of the field of scalars. This was expounded upon by Breaz et al.~\cite{Breaz:2024}. A gauge or frame-independent formulation of analytical dynamics requires affine bundles and, because of this, Grabowska et al.\ (see~\cite{Grabowska:2003}) defined Lie brackets on sections of affine bundles. Brzezi\'nski \cite{Brzezinski:2022} developed a frame-independent reformulation of Lie brackets on affine bundles using ternary structures.

Generically, we view algebraic structures based on heaps and closely related ternary operations as `affine algebraic structures' in the sense that some privileged element is not fixed. Thus, developing ternary algebraic structures is seen as part of developing structures that do not rely on an underlying vector space but are, in some sense, affine in nature. Developing such ternary structures offers a new perspective on binary algebraic structures, affine structures, and potentially physics.

\subsection{Preliminaries} We quickly review the basic theory of heaps and trusses as needed later. Our main reference for heaps and related structures is the book by Hollings and Lawson~\cite{Hollings:2017}. For categorical notions, we refer to Mac~Lane \cite{MacLane:1988}.

\begin{Definition}[\cite{Baer:1929,Prufer:1924}]
A \emph{heap} is a set $H$, together with a ternary operation
\[ [-,-,-]\colon \ H \times H \times H \rightarrow H,\]
that satisfies
\begin{align}\label{eq:HeapCons}
 [[a,b,c], d, e] = [a,b,[c,d,e]], \qquad [a,b,b] = [b,b, a] =a,
\end{align}
for all $a$, $b$, $c$, $d$ and $e \in H$. The second equation of \eqref{eq:HeapCons} is known as the Malt'cev identities. A~\emph{morphism of heaps} is a map $\phi\colon H \rightarrow H'$ that preserves the ternary operation, i.e.,
\[\phi [a,b,c]_{H} = [\phi(a), \phi(b), \phi(c)]_{H'}.\]
A heap is said to be an \emph{abelian heap} if for all $a$, $b$ and $c \in H$
\[ [a,b,c] = [c,b,a].\]
\end{Definition}

 We will denote the category of heaps as $\catname{Heap}$, and the category of abelian heaps as $\catname{AbHeap}$. Note that for any heap, we can also deduce the equality
\begin{equation*}
[[a,b,c], d,e] = [a, [d,c,b],e] = [a,b, [c,d,e]],
\end{equation*}
which is referred to as para-associativity. Other useful results include the following. If ${[a,b,c]\!=\! d}$, then $a= [d,c,b]$, and $a=b$ if and only if $[a,b,c]=c$, for all $a$, $b$ and $c\in H$. Moreover, for abelian heaps, we have the \emph{transposition rule} \cite[Lemma~2.3]{Brzezinski:2020}
\begin{equation*}
[[a_1, a_2, a_3], [b_1, b_2, b_3], [c_1, c_2, c_3]] = [[a_1, b_1, c_1], [a_2, b_2, c_2], [a_3, b_3, c_3]].
\end{equation*}

\begin{Example}\label{exa:HeapG}
Given any group $G$, it can be `heapified' by defining $[g_1, g_2, g_3] := g_1 g_2^{-1} g_3$.
\end{Example}

Recall that given any heap $(H, [-,-,-])$, by selecting an element $b\in H$, we can `groupify' the heap by defining $a\cdot c := [a,b,c]$ and $a^{-1} = [b,a,b]$. For an abelian heap, we will use additive notation $+_b$ and $-_b$. It is a well-known fact that all heaps arise via the `heapification' of a~group. However, we do not have an equivalence of categories, as the selection of an element to `groupify' is not natural. However, there is an isomorphism of categories $\catname{Heap}_* \cong \catname{Grp}$, where we consider pointed heaps and homomorphisms that preserve the privileged points.

\begin{Definition}[\cite{Brzezinski:2019}]\label{def:Truss}
A \emph{truss} is an abelian heap $T$ together with an associative binary operation (multiplication) that distributes over the heap operation, i.e.,
\[ s[t_1, t_2, t_3] = [st_1, st_2, st_3]  \qquad \textnormal{and} \qquad [t_1, t_2 , t_3]s = [t_1 s, t_2 s, t_3 s]\]
for all $t_1, t_2, t_3$ and $s \in T$.

A truss is said to be a \emph{unital truss} if there is a neutral element for the multiplication. An element $0 \in T$ is called an \emph{absorber} or \emph{zero} if $0t = 0 = t 0$, for all $t \in T$.

A \emph{truss morphism} is a heap morphism that respects the multiplication. The resulting category we denote as $\catname{Truss}$. Morphisms between unital trusses respects the units, and morphisms between trusses with zeros respects the zeros. We will denote the subcategories of unital trusses and trusses with a zero as $\catname{Truss}_1$ and $\catname{Truss}_0$, respectively.
\end{Definition}

Clearly, we have two forgetful functors
\begin{equation}\label{eqn:ForFunc}
U \colon \  \catname{Truss} \rightarrow \catname{AbHeap}, \qquad V \colon \ \catname{Truss} \rightarrow \catname{SemiGrp},
\end{equation}
from the category of trusses to the categories of abelian heaps and semigroups, respectively.
 We define $s+_t u: = [s,t, u]$ and $-_t u := [t,u,t]$, and then following Brzezi\'nski \cite[Lemma~3.9]{Brzezinski:2020} we have
\begin{enumerate}\itemsep=0pt
\item If $T$ is unital, then $+_1$ together with the binary product satisfy a brace-type distributivity rule, i.e., for all $s$,~$t$ and $u \in T$,
\[s(t+_1 u) = st -_1 s +_1 su, \qquad (t+_1 u)s = ts -_1s +_1 us.\]
\item If $0 \in T$ is an absorber, then $+_0$ together with the binary product defines a ring structure on the set $T$.
\end{enumerate}
From Brzezi\'nski \cite[Corollary 3.10]{Brzezinski:2020}, a unital truss is \emph{braceable} if the binary product defines a~group. That is, $(T, +_1, \cdot)$ is a two-sided brace in the sense of Rump. Conversely, every two-sided brace gives rise to a unital truss where the binary product defines a group.

We denote the category of braceable trusses as $\catname{Truss}_{\catname{Brac}}$ and the category of two-sided braces as $\catname{2SBrac}$. Observe that we have the ``braceification'' and ``ringification'' functors
\begin{subequations}
\begin{align}\label{eqn:BrFunc}
& \mathrm{Br}\colon \ \catname{Truss}_{\catname{Brac}}\rightarrow \catname{2SBrac},\\ \label{eqn:RnFunc}
& \mathrm{Rn}\colon \ \catname{Truss}_0 \rightarrow \catname{Rng},
\end{align}
\end{subequations}
respectively. Homomorphisms remain unchanged and are just considered in a different category.

We also have a ``trussification'' functor $\mathrm{Tr} \colon \catname{Rng} \rightarrow \catname{Truss}_0 $ by defining $[s,t,u] := s-t+u$ (see Example \ref{exa:HeapG}) and using the ring product. We have the following important result.

\begin{Lemma}\label{lem:Trussification}
The categories $\catname{Rng}$ of rings and $\catname{Truss}_0$ of trusses with an absorber are isomorphic.
\end{Lemma}

\begin{proof}[Sketch of proof] It is well known that there is an isomorphism of categories between $\catname{AbGrp}$ and $\catname{AbHeap}_*$ (pointed abelian heaps). As the binary product remains unchanged, this isomorphism lifts to $\catname{Rng}$ and $\catname{Truss}_0$.
\end{proof}

Rump's theorem (see \cite{Rump:2007}) and its generalisations state that solutions to the set-theoretic Yang--Baxter equation can be built from a brace. For a map $r\colon T \times T \rightarrow T \times T$, where we have $r(s, t) = (\lambda_s(t), \rho_t(s))$, the braid relation form of the \emph{set-theoretic Yang--Baxter equation} (see~\cite{Drinfeld:1992}) is
\begin{equation}\label{eqn:SetYB}
(r \times \Id_T)(\Id_T \times r)(r \times \Id_T) = (\Id_T \times r)(r \times \Id_T)(\Id_T \times r).
\end{equation}
Any such map $r$ that satisfies \eqref{eqn:SetYB} is said to be a \emph{Yang--Baxter map} following Veselov \cite{Veselov:2003} (also see Buchstaber \cite{Buchstaber:1998}). If $\lambda_s$ and $\rho_s$ are bijections, then~$r$ is said to be a \emph{non-degenerate solution}. Importantly, the maps $\lambda$ and $\rho$ can be built from the two-sided brace. It is well known that for~$r$ to be a solution (see, for example,~\cite{Doikou:2024}), we require
\begin{subequations}
\begin{align}\label{eqn:YB1}
& \lambda_s \big(\lambda_t(u) \big) = \lambda_{\lambda_s(t)}\big( \lambda_{\rho_t(s)}(u)\big) ,\\ \label{eqn:YB2}
& \rho_u\big(\rho_t(s) \big) = \rho_{\rho_u(t)}\big(\rho_{\lambda_t(u)}(s) \big),\\ \label{eqn:YB3}
& \lambda_{\rho_{\lambda_t(u)}(s)}\big(\rho_u(t) \big) = \rho_{\lambda_{\rho_s}(u)}\big( \lambda_s(t)\big)
\end{align}
\end{subequations}
for all $s$, $t$ and $u \in T$.

We will assume the reader has a working knowledge of $\Z_2$-graded algebra. The Grassmann parity of an object will be denoted by `tilde', i.e., $\widetilde{O} \in \{\mathbf{0},\mathbf{1} \} = \Z_2$. Recall that an associative supercommutative superalgebra is an associative algebra (over a unital commutative ring $\mathbb{K}$) that has a decomposition $A = A_\mathbf{0} \oplus A_\mathbf{1}$, such that $A_i A_j \subseteq A_{i+j}$ (read mod~$2$), and \smash{$ab = (-1)^{\widetilde{a} \widetilde{b}} ba$} for all $a,b \in A$. For brevity, we refer to superalgebras.

\begin{Example}
The Grassmann algebra $\mathbb{K}\bigl[\theta^1, \dots , \theta^n\bigr]$, is the algebra of polynomials in formal generators $\theta^i$, subject to $\theta^i \theta^j = - \theta^j \theta^i$, which implies \smash{$\bigl(\theta^i\bigr)^2 =0$}. $\mathbb{K}_\mathbf{0}\bigl[\theta^1, \dots , \theta^n\bigr]$ consists of polynomials of even degree in~$\theta$s, and $\mathbb{K}_\mathbf{1}\bigl[\theta^1, \dots , \theta^n\bigr]$ consists of polynomials of odd degree in~$\theta$s.
\end{Example}

\begin{Example}
Similarly to the previous example, the polynomial algebra
\[\mathbb{K}\bigl[x^1, \dots , x^m, \theta^1, \dots , \theta^n\bigr] = \mathbb{K}\bigl[x^1, \dots, x^m\bigr]\otimes_\mathbb{K} \mathbb{K}\bigl[\theta^1, \dots , \theta^n\bigr],\]
where $\widetilde{x}^a = 0$ and $\widetilde{\theta}^i =1$, subject to the relations $x^a x^b = x^b x^a$, $x^a \theta^i = \theta^i x^a$, and $\theta^i \theta^j = - \theta^j \theta^i$, is a superalgebra.
\end{Example}

\begin{Example}
The cohomology ring of a topological space (via the cup product) is a superalgebra.
\end{Example}

Given two superalgebras $A$ and $B$, the tensor product $A \otimes B$ is an algebra with the multiplication being defined as \smash{$(a_1 \otimes b_1)(a_2 \otimes b_2) := (-1)^{\widetilde{b}_1 \widetilde{a}_2} (a_1 a_2 \otimes b_1 b_2)$}. The reader may consult \cite[Section~1.1]{Carmeli:2011} for details.

We will not employ the machinery of sheaves and superspaces explicitly in this paper. For an introduction to superalgebras and supergeometry, the reader may consult Carmeli et al.\ \cite{Carmeli:2011}, Leites \cite{Leites:1980}, Manin \cite{Manin:1997}, and Varadrajan \cite{Varadrajan:2004}, for example. We remark that Lie superheaps were defined by the author in~\cite{Bruce:2025}.

\section{Affine supertrusses, their morphisms and examples}
\subsection{Affine supertrusses}
For the various categorical notions employed in this work, the reader may consult Mac Lane~\cite{MacLane:1988}. Our reference for coalgebras is Brzezi\'nski and Wisbauer \cite[Chapter~1]{Brezinski:2003}. Let $\mathbb{K}$ be a unital commutative ring. We denote the category of unital associative supercommutative superalgebras over $\mathbb{K}$ as $\catname{SAlg}_{\mathbb{K}}$. Recall that morphisms of superalgebras preserve the $\Z_2$-grading (Grassmann parity).
\begin{Definition}\label{def:AffSupTrus}
An \emph{affine supertruss $($scheme$)$} $T$ is a representable functor
\[T \colon \ \catname{SAlg}_{\mathbb{K}} \rightarrow \catname{Truss}.\]
A \emph{morphism of affine supertrusses} is a natural transformation $\phi_- \colon T(-) \rightarrow T'(-)$. The resulting (functor) \emph{category of affine supertrusses} is defined as $\catname{ASTru}_{\mathbb{K}} := [\catname{SAlg}_{\mathbb{K}}, \catname{Truss}]_{\mathrm{rep}}$, where we restrict to the full subcategory of $[\catname{SAlg}_{\mathbb{K}}, \catname{Truss}]$ consisting of representable functors.
\end{Definition}

Decoding this definition, given any superalgebra $A \in \catname{SAlg}_{\mathbb{K}}$, then $T(A)$ is a truss. Representable means that $T(-) \cong h^X(-) := \Hom_{\catname{SAlg}_{\mathbb{K}}}(X, -)$ for some superalgebra $X \in \catname{SAlg}_{\mathbb{K}}$. A~representation is a choice of a superalgebra $X$ and a natural isomorphism. With some minor abuse of language, we may refer to $X$ as \emph{the} representing superalgebra.
Given two superalgebras $A, B \in \catname{SAlg}_{\mathbb{K}}$ and a superalgebra homomorphism
$\psi \in \Hom_{\catname{SAlg}_{\mathbb{K}}}(A,B)$,
we have a truss homomorphism
\[T(\psi) \colon \ T(A) \rightarrow T(B).\]
Explicitly, for any $s \in T(A)$, we compose with $\psi$ to obtain $\psi \circ s \in T(B)$.

The representing superalgebra $X$ must, of course, have the structure of a ``cotruss'' to ensure~$h^X$ lands in the category of trusses rather than just sets. In particular, $X$ is required to admit a pair of superalgebra homomorphisms\footnote{As we are concerned with algebraic constructions, the tensor products in this note are algebraic, and we do not consider topological completions.}
 \begin{equation*}
 \Delta^{(2)}\colon \ X \rightarrow X \otimes X, \qquad \Delta^{(3)}\colon \ X \rightarrow X \otimes X \otimes X,
 \end{equation*}
which we refer to as the \emph{binary comultiplication} and \emph{ternary coaddition}, respectively. Note that the tensor product (over $\mathbb{K}$) is the $\Z_2$-graded tensor product. We require $\bigl(X, \Delta^{(3)}\bigr)$ to be abelian quantum heap (see \v{S}koda \cite{Skoda:2007}), thus
\begin{subequations}\label{eqn:Con}
\begin{align}\label{eqn:Con1}
& \bigl(\Id \otimes \Id \otimes \Delta^{(3)}\bigr)\circ \Delta^{(3)} = \bigl(\Delta^{(3)} \otimes \Id \otimes \Id\bigr)\circ \Delta^{(3)},\\ \label{eqn:Con2}
& (\Id \otimes m_X) \circ \Delta^{(3)} = \Id \otimes !_X , \\ \label{eqn:Con3}
&(m_X \otimes \Id) \circ \Delta^{(3)} = !_X \otimes \Id, \\ \label{eqn:Con4}
& \Delta^{(3)} = \sigma_{13}\circ \Delta^{(3)},
\end{align}
where $\Id\colon X \rightarrow X$ is the identity map, and $!_X$ denotes the unit mapping $\mathbb{K} \rightarrow X$ given by~$k \mapsto k 1_X$. Here $\sigma_{13}$ interchanges elements in first and third positions, i.e., for a homogeneous element $\sigma_{13}(x_1 \otimes x_2 \otimes x_3 ) := x_3 \otimes x_2 \otimes x_1$. The pair $\bigl(X, \Delta^{(2)}\bigr)$ must be a (non-counital) coassociative coalgebra, thus
\begin{equation}\label{eqn:Con5}
\bigl(\Delta^{(2)}\otimes \Id \bigr) \circ \Delta^{(2)} = \bigl(\Id \otimes \Delta^{(2)}\bigr)\circ \Delta^{(2)}.
\end{equation}
The compatibility conditions between the two structures are left and right co-distributivity
\begin{align}\label{eqn:Con6}
&\bigl(\Id \otimes \Delta^{(3)}\bigr)\circ \Delta^{(2)} = m_X^{135} \circ \bigl(\Delta^{(2)}\otimes \Delta^{(2)}\otimes \Delta^{(2)} \bigr) \circ \Delta^{(3)}, \\ \label{eqn:Con7}
&\bigl(\Delta^{(3)} \otimes \Id \bigr)\circ \Delta^{(2)} = m_X^{246}\circ\bigl(\Delta^{(2)}\otimes \Delta^{(2)}\otimes \Delta^{(2)} \bigr) \circ \Delta^{(3)},
\end{align}
\end{subequations}
where, acting on a basic homogeneous element $x = x_1 \otimes x_2 \otimes \cdots \otimes x_6$, we have defined
\begin{align*}
&m_X^{135}(x) = (-1)^{\varepsilon} (x_1 x_3 x_5) \otimes x_2 \otimes x_4 \otimes x_6, \\
&m_X^{246}(x) = (-1)^{\varepsilon} x_1 \otimes x_3 \otimes x_5 \otimes (x_2 x_4 x_6),
\end{align*}
where the sign factor is $\varepsilon := \widetilde{x_2}\widetilde{x_3} + \widetilde{x_5} \widetilde{x_4} + \widetilde{x_5}\widetilde{x_2}$.

The connection with the truss operations is as follows. Consider $t_1, t_2, t_3 \in T(A)$ for some $A \in \catname{SAlg}_{\mathbb{K}}$, then
\begin{align*}
& X \xrightarrow{\Delta^{(2)}} X \otimes X \xrightarrow{(t_1\otimes t_2)} A \otimes A \xrightarrow{m^{(2)}_A} A,\\ \nonumber
& X \xrightarrow{\Delta^{(3)}} X \otimes X \otimes X \xrightarrow{(t_1 \otimes t_2 \otimes t_3)} A \otimes A \otimes A \xrightarrow{m^{(3)}_A} A,
\end{align*}
describe the binary product and the abelian heap operation. Here we employ the notation $m^{(i)}_A\colon A^{\otimes i} \rightarrow A$ as the (ordered) multiplication of $i$ elements. More explicitly, we write
\begin{align}\label{eqn:BiProdTerAdd}
 st := m^{(2)}_A \circ (s \otimes t) \circ \Delta^{(2)},
\qquad [t_1, t_2, t_3] := m^{(3)}_A \circ (t_1 \otimes t_2 \otimes t_3) \circ \Delta^{(3)}.
\end{align}
If $X$ admits a counit, then the resulting truss $T(A)$ is unital.

Morphisms of affine supertrusses are defined as natural transformations. This means that for any $A\in \catname{SAlg}_\mathbb{K}$, the component $\phi_A\colon T(A) \rightarrow T'(A)$ is a truss morphism (see Definition~\ref{def:Truss}), i.e., a heap morphism that respects the binary multiplication (and units if we restrict to unital affine supertrusses).

\begin{Theorem}\label{trm:RepObs}
Let $T \in \catname{ASTru}_{\mathbb{K}}$ be an affine supertruss, then the representing object is a triple $\bigl(X, \Delta^{(2)}, \Delta^{(3)}\bigr)$, where $X \in \catname{SAlg}_{\mathbb{K}}$, and
\[\Delta^{(2)} \colon \ X \rightarrow X \otimes X, \qquad \Delta^{(3)}\colon \ X \rightarrow X \otimes X \otimes X,\]
are superalgebra homomorphisms that satisfy the conditions \eqref{eqn:Con}.
\end{Theorem}

\begin{proof}
Take some arbitrary $A \in \catname{SAlg}_{\mathbb{K}}$ and consider the set of superalgebra homomor\-phisms $T(A)$. We will freely use~\eqref{eqn:BiProdTerAdd}. From the definition of the target category, we have three parts to check.

{\bf $\boldsymbol{\bigl(X, \Delta^{(3)}\bigr)}$ is an abelian quantum heap.} The functor $T(-) := \Hom_{\catname{SAlg}_\mathbb{K}}(X,-)$, via a~forgetful functor, takes its values in the category of abelian heaps. Thus, $[-,-,-]_A\colon T(A) \times T(A) \times T(A) \rightarrow T(A)$ defines a natural transformation $h^{X \otimes X \otimes X}(-) \rightarrow h^X(-)$. By Yoneda's lemma,\footnote{For details of Yoneda's lemma, see \cite[Chapter~III, Section~2]{MacLane:1988}.} we have a unique superalgebra homomorphism $\Delta^{(3)}\colon X \rightarrow X \otimes X \otimes X$. Then, following \v{S}koda \cite{Skoda:2007}, conditions \eqref{eqn:Con1}--\eqref{eqn:Con4} on $\Delta^{(3)}$ directly follow from $T(A)$ being an abelian heap. In more detail, directly expanding the expressions we obtain
\begin{align*}
& [[t_1, t_2, t_3] , t_4, t_5] = m_A^{(5)} \circ(t_1\otimes t_2 \otimes t_3 \otimes t_4 \otimes t_5 ) \circ\bigl(\Delta^{(3)} \otimes \Id \otimes \Id\bigr) \circ \Delta^{(3)}, \\
& [t_1, t_2, [t_3 , t_4, t_5]] = m_A^{(5)} \circ(t_1\otimes t_2 \otimes t_3 \otimes t_4 \otimes t_5 ) \circ\bigl( \Id \otimes \Id \otimes \Delta^{(3)}\bigr) \circ \Delta^{(3)}.
\end{align*}
Then, via Yoneda's lemma, these must be equal, which gives \eqref{eqn:Con1}. Continuing, we have
\begin{align*}
& [s, t, t ] = m_A^{(3)} \circ (s\otimes t \otimes t) \circ \Delta^{(3)} = m_A^{(3)} \circ(s \otimes t) \circ (\Id \otimes m_X)\circ \Delta^{(3)},\\
& s = m_A \circ(s \otimes t) \circ (\Id \otimes !_X) ,
\end{align*}
noting that $t(1_X) = 1_A$. Then, via Yoneda's lemma, these must be equal, which gives~\eqref{eqn:Con2}. The condition \eqref{eqn:Con3} similarly follows. Finally condition \eqref{eqn:Con4} immediately follows from $[t_1, t_2 , t_3] = [t_3, t_2, t_1]$ and Yoneda's lemma.

{\bf $\boldsymbol{\bigl(X , \Delta^{(2)}\bigr)}$ is a coassociative coalgebra.} The functor $T(-) := \Hom_{\catname{SAlg}_\mathbb{K}}(X,-)$, via a~forgetful functor, takes its values in the category of semigroups. Thus, $\cdot_A\colon T(A) \times T(A) \rightarrow T(A)$ defines a natural transformation $h^{X \otimes X }(-) \rightarrow h^X(-)$. By Yoneda's lemma, we have a unique superalgebra homomorphism $\Delta^{(2)}\colon X \rightarrow X \otimes X$. From the standard theory of coalgebras, we observe that if $T(A)$ has the structure of a semigroup, then~$\Delta^{(2)}$ satisfies~\eqref{eqn:Con5}.

{\bf Compatibility of $\boldsymbol{\Delta^{(2)}}$ and $\boldsymbol{\Delta^{(3)}}$.} The only non-standard conditions to check are \eqref{eqn:Con6}--\eqref{eqn:Con7} that describe distributivity. We will establish the consequences of left distributivity only, as the right follows from near identical calculations.
Using \eqref{eqn:BiProdTerAdd}, we have
\begin{align*}
s [t_1,t_2,t_3] & = m^{(2)}_A \circ \bigl( s \otimes \bigl( m^{(3)}_A \circ (t_1\otimes t_2 \otimes t_3)\circ \Delta^{(3)} \bigr) \bigr) \circ \Delta^{(2)}\\
&= m^{(4)}_A \circ (s \otimes t_1 \otimes t_2 \otimes t_3) \circ \bigl(\Id \otimes \Delta^{(3)}\bigr)\circ \Delta^{(2)}.
\end{align*}
On the other hand,
\begin{align*}
[st_1, st_2, s t_3] = m^{(3)}_A \circ \bigl(
&\bigl(m^{(2)}_A \circ (s \otimes t_1)\circ \Delta^{(2)} \bigr)
 \otimes \bigl(m^{(2)}_A \circ (s \otimes t_2)\circ \Delta^{(2)} \bigr) \\
&\otimes \bigl(m^{(2)}_A \circ (s \otimes t_3)\circ \Delta^{(2)} \bigr)\bigr)\circ \Delta^{(3)}.
\end{align*}
As $A$ is a supercommutative superalgebra and $s \in \Hom_{\catname{SAlg}_{\mathbb{K}}}(X, A)$ is a superalgebra homomorphism, we can factor out the copies of $s$ and group their multiplication in $A$. This is equivalent to applying $s$ to a product of elements in $X$. Specifically, we multiply $s$ with the first, third, and fifth tensor factors generated by $\Delta^{(2)} \otimes \Delta^{(2)} \otimes \Delta^{(2)}$. This grouped multiplication is exactly the map defined as $m_X^{135}$ (assuming homogeneity). Thus, we obtain
\[[st_1, st_2, s t_3] = m^{(4)}_A \circ (s \otimes t_1 \otimes t_2 \otimes t_3 )\circ m^{135}_X \circ \bigl(\Delta^{(2)}\otimes \Delta^{(2)} \otimes \Delta^{(2)}\bigr)\circ \Delta^{(3)}.\]
Insisting on left distributivity requires $\bigl(\Id \otimes \Delta^{(3)}\bigr)\circ \Delta^{(2)} = m_X^{135} \circ \bigl(\Delta^{(2)}\otimes \Delta^{(2)}\otimes \Delta^{(2)} \bigr) \circ \Delta^{(3)}$, which is precisely~\eqref{eqn:Con6}.
\end{proof}

Selecting privileged elements of an affine supertruss, such as a unit or absorber, requires care. First, recall that for every superalgebra, there is a unique superalgebra homomorphism (the canonical inclusion)
$!_A \colon  \mathbb{K} \hookrightarrow A $,
defined as $\mathbb{K} \ni k1 \mapsto k 1_A \in A$. Moreover, given a~$\psi \in \Hom_{\catname{SALg}_{\mathbb{K}}}(A, B)$ we have $\psi \circ !_A = !_B$. Specifying a superalgebra homomorphism $\varepsilon\colon X \rightarrow \mathbb{K}$ defines an element $!_A \circ \varepsilon \in T(A)$. In other words, $T(A)$ is canonically a pointed set. Under $T(\psi)\colon T(A) \rightarrow T(B)$, we observe that
\[ !_A \circ \varepsilon \mapsto \psi \circ !_A \circ \varepsilon = !_B \circ \varepsilon.\]
Thus, we have a well-behaved notion of ``selecting elements'' of~$T$.

A \emph{counit} is a superalgebra homomorphism $\varepsilon_{\textrm{unit}}\colon X \rightarrow \mathbb{K}$ that satisfy the condition
\begin{equation}\label{eqn:Counit}
(\varepsilon_{\textrm{unit}} \otimes \Id )\circ \Delta^{(2)} = \Id = (\Id \otimes \varepsilon_{\textrm{unit}}) \circ \Delta^{(2)}.
\end{equation}
If a counit exists, then every truss $T(A)$ is a unital truss with $e := !_A \circ \varepsilon_{\textrm{unit}}$. We will refer to an affine supertruss such that the representing superalgebra has a counit as an \emph{affine unital supertruss}.

A \emph{cozero} (or \emph{coabsorber}) is a superalgebra homomorphism $\varepsilon_{\textrm{zero}}: X \rightarrow \mathbb{K}$ that satisfy the condition
\begin{equation}\label{eqn:Cozero}
(\Id \otimes \varepsilon_{\textrm{zero}} )\circ \Delta^{(2)} = !_X \circ \varepsilon_{\textrm{zero}} = (\varepsilon_{\textrm{zero}} \otimes \Id ) \circ \Delta^{(2)}.
\end{equation}
If a cozero exists, then every truss $T(A)$ has a zero element for the binary product $z := !_A \circ \varepsilon_{\textrm{zero}}$.

If the affine supertrusses are equipped with a unit and/or a zero, then components of the homomorphisms are required to respect these structures. To be careful, for any $A \in \catname{SAlg}_{\mathbb{K}}$ and a natural transformation $\phi_- \colon T(-) \rightarrow T'(-)$, we have
\begin{equation*}
\phi_A [s,t,u]_T = [\phi_A(s), \phi_A(t), \phi_A(u)]_{T'}, \qquad \phi_A(st) = \phi_A(s)\phi(t),
\end{equation*}
for all $s$,~$t$ and $u \in T(A)$. Insisting on preserving units and zero means
\begin{equation*}
\phi_A(e) = e', \qquad \phi_A(z) = z'.
\end{equation*}
We define, as standard, $J_X$ as the ideal in $X$ generated by the odd elements of~$X$. The reduced algebra is then $X_{\mathrm{red}}:= X / J_X$. Given an affine supertruss $T$ and its representing superalgbera~$X$, then the \emph{reduced affine truss} is defined (up to isomorphism) as
\[T_{\mathrm{red}}(-) \cong \Hom_{\catname{SAlg}_{\mathbb{K}}}(X_{\mathrm{red}}, -),\]
together with the restricted maps $\Delta^{(2)}_{\mathrm{red}} := \Delta^{(2)}|_{X_{\mathrm{red}}}$ and $\Delta^{(3)}_{\mathrm{red}} := \Delta^{(3)}|_{X_{\mathrm{red}}}$. By considering $\catname{Alg}_\mathbb{K}$ as a full subcategory of $\catname{SAlg}_{\mathbb{K}}$, the reduced affine truss is a truss object in the category of affine schemes.

We can recover classical trusses by considering $T(\mathbb{K}) \cong \Hom_{\catname{SAlg}_\mathbb{K}}(X, \mathbb{K})$. That is, elements of~$T(\mathbb{K})$ are in one-to-one correspondence with the $\mathbb{K}$-points of the affine scheme defined by $X_{\mathrm{red}}$.

\begin{Definition}
Let $T \in \catname{ASTru}_{\mathbb{K}}$ be an affine supertruss. The \emph{underlying truss} of $T$ is defined~as
\[ T_\mathbb{K} := T_\mathrm{red}(\mathbb{K}) \cong \Hom_{\catname{SAlg}_\mathbb{K}}(X_\mathrm{red} , \mathbb{K})\in \catname{Truss}.\]
\end{Definition}

The forgetful functors \eqref{eqn:ForFunc} extend to affine supertrusses in the obvious way.

\begin{Proposition}
There exist forgetful functors
\begin{align*}
& \mathcal{U}\colon \ \catname{ASTrus}_\mathbb{K} \rightarrow [\catname{SAlg}_\mathbb{K}, \catname{AbHeap}]_{\mathrm{rep}},
\qquad \mathcal{V}\colon \  \catname{ASTrus}_\mathbb{K} \rightarrow [\catname{SAlg}_\mathbb{K}, \catname{SemiGrp}]_{\mathrm{rep}}.
\end{align*}
\end{Proposition}

\begin{proof}
Using the forgetful functors \eqref{eqn:ForFunc} we define for any affine supertruss $T$
\[ \mathcal{U}(T) := U \circ T , \qquad \mathcal{V}(T) := V\circ T.\]
As we have the composition of functors, we have well-defined functors. As the functors $U$ and $V$ are forgetful, they do not affect the representability of the composition with an affine supertruss.
\end{proof}

\subsection{Abelian supergroups as affine supertrusses}
An affine abelian supergroup is a representable functor
\[
G\colon \ \catname{SAlg}_{\mathbb{K}} \longrightarrow \catname{AbGrp},
\]
its representing object is a Hopf superalgebra (see \cite[Chapter~11]{Carmeli:2011}). Let
\[
H \colon\ \catname{AbGrp} \longrightarrow \catname{Truss},
\]
be the `trussification' functor ($H$ comes from `heapification'). From a (set-theoretical) group, it assigns the abelian heap operation $[a,b,c] := a-b+c$ and the binary product is the group product $ab := a+b$, where we write the group as additive. Observe that $a[b,c,d] = a+b -c+d = (a+b)- (a+c)+ (a+d)$ is precisely the truss distribution rule. As the truss structure is built from the group structure alone, group homomorphisms are also truss homomorphisms.

Given an affine abelian supergroup, we define the associated affine supertruss as the composition of functors
\begin{equation*}
T := H \circ G \colon \ \catname{SAlg}_{\mathbb{K}} \longrightarrow \catname{Truss}.
\end{equation*}
However, we must check the naturality of the truss operations. Given some superalgebra homomorphism $\psi\colon A \rightarrow B$, $\psi^G := G(\psi)$ is a group homomorphism $\psi^G\colon G(A) \rightarrow G(B)$. The reader can quickly check that $\psi$ induces a truss homomorphism, and the construction of the truss operations is natural.

Representability follows from the Hopf superalgebra representing the affine supergroup. In particular, we have the comultiplication $\Delta\colon X \rightarrow X \otimes X$ and the antipode $S\colon X \rightarrow X$. We then define the representing superalgebra as $\bigl(X, \Delta^{(2)} , \Delta^{(3)} \bigr)$, where
\[
\Delta^{(2)} := \Delta, \qquad \Delta^{(3)} := (\Id \otimes S \otimes \Id )\circ \Delta_2,\]
where $\Delta_2$ is the iterated comultiplication, i.e., $\Delta_2:= (\Delta \otimes \Id)\circ \Delta$. We leave it as an exercise to the reader to check that the conditions of Theorem~\ref{trm:RepObs} hold.

\begin{Remark}
Although we chose the binary product as simply the group product, other choices are possible, including those specific to the group structure. Some `canonical' choices here include
\begin{enumerate}[(i)]\itemsep=0pt
\item left-zero semigroup $st := s$;
\item right-zero semigroup $st := t$; and
\item trivial semigroup $st :=0$,
\end{enumerate}
defined for $s, t, 0 \in G(A)$, where $0$ is the identity element. Importantly, these semigroup structures are natural and so can be used to define affine supertruss structures given an affine abelian supergroup.
\end{Remark}
\begin{Example}\label{exa:TransSupG}
The affine supergroup $\mathbb{G}_a^{1|1}$ of translations of $\mathbb{K}^{1|1}$ is represented by the Hopf superalgebra $X \cong \mathbb{K}[x, \theta]$, with $\widetilde{x} = 0$ and $\widetilde{\theta} =1$, subject to $\theta^2 =0$, with the maps $\Delta(x) := x \otimes 1 + 1 \otimes x $, $\Delta(\theta) := \theta \otimes 1 + 1 \otimes \theta$, $\epsilon(x):= \epsilon(\theta) :=0$, $S(x) := - x$ and $S(\theta) := - \theta$. Then ``trussifying'' we obtain
\begin{align*}
& \Delta^{(2)}(x) := x \otimes 1 + 1 \otimes x, \\
& \Delta^{(2)}(\theta) := \theta \otimes 1 + 1 \otimes \theta,\\
& \Delta^{(3)}(x) := x \otimes 1 \otimes 1 - 1 \otimes x \otimes 1 + 1 \otimes 1 \otimes x, \\
&\Delta^{(3)}(\theta) := \theta \otimes 1 \otimes 1 - 1 \otimes \theta \otimes 1 + 1 \otimes 1 \otimes \theta.
\end{align*}
\end{Example}

\subsection{Finite projective superrings as affine supertrusses}
 We will consider superrings over a unital commutative ring $\mathbb{K}$. That is, we have superrings that are also $\mathbb{K}$-algebras. A~superring $R \in \catname{SRng}_\mathbb{K}$ can be considered as a functor $R(-)\colon \catname{SAlg}_\mathbb{K} \rightarrow \catname{Rng}$, via the `even rules' (see \cite[Section~1.7]{Deligne:1999} and/or \cite[Section~1.8]{Carmeli:2011}). That is, we define the functor~as
\begin{description}
\item[Objects] $R(A) := \big(R \otimes_{\mathbb{K}} A\big)_\mathbf{0}$, for all $A \in \catname{SAlg}_\mathbb{K}$;
\item[Morphisms] $R(\psi):= \Id_R \otimes \psi$, where $\psi : A \rightarrow B$.
\end{description}
Note that $R(A) = \big(R_\mathbf{0} \otimes_\mathbb{K} A_\mathbf{0}\big) \oplus \big(R_\mathbf{1} \otimes_\mathbb{K} A_\mathbf{1}\big)$. We will write basic elements of $R(A)$ as $s = (r \otimes a , \varrho \otimes \alpha)$. The binary multiplication is then given by
\begin{equation*}
st := (r_1 r_2 \otimes a_1 a_2 - \varrho_1 \varrho_2 \otimes \alpha_1 \alpha_2, r_1 \varrho_2 \otimes a_1 \alpha_2 + \varrho_1 r_2 \otimes \alpha_1 a_2).
\end{equation*}
We can compose with the `trussification' functor (see Lemma \ref{lem:Trussification}) to obtain the functor
\[\big(Tr \circ R \big)(-)\colon \ \catname{SAlg}_\mathbb{K} \rightarrow \catname{Truss}_\mathbf{0},\]
which leads to the abelian heap operation
\begin{equation*}
[s,t,u] = (r_1 \otimes a_1 - r_2 \otimes a_2 + r_3 \otimes a_3 , \varrho_1 \otimes \alpha_1 - \varrho_2 \otimes \alpha_2 + \varrho_3 \otimes \alpha_3).
\end{equation*}
Representability of the constructed functor is not automatic. We define the dual
\[R^* := \Hom_{\catname{SMod}_\mathbb{K}}(R, \mathbb{K}),\]
and observe that $X:= \mathrm{Sym}(R^*)$ (the free supersymmetric algebra) has the structure of a supercoalgebra. Algebraically, the dual and free supersymmetric algebra are well-defined. However, for the fundamental isomorphism
\begin{align*}
  (R\otimes_\mathbb{K} A)_\mathbf{0} \cong \Hom_{\catname{SMod}}(R^*, A),
 \qquad r \otimes a \longmapsto (\lambda \mapsto \lambda(r) a),
\end{align*}
we require $R$ to be a finitely generated projective $\mathbb{K}$-module. When $\mathbb{K}$ is a field, this is equivalent to $R$ being finite-dimensional as a vector space. This is unavoidable if we restrict attention to algebraic constructions and avoid functional analysis. We identify $R(A)$ with $\mathbb{K}$-linear grading-preserving maps from $R^*$ to $A$. In particular, maps $R^*_\mathbf{0} \rightarrow A_\mathbf{0}$ correspond to elements of $R_\mathbf{0} \otimes_\mathbb{K} A_\mathbf{0}$, and maps $R^*_\mathbf{1} \rightarrow A_\mathbf{1}$ correspond to elements of $R_\mathbf{1} \otimes_\mathbb{K} A_\mathbf{1}$.

On $X := \mathrm{Sym}(R^*)$, we have $\Delta_+$ and $S$, which are the comultiplication and antipode associated with the additive structure on $R$. We furthermore have $\Delta_\times$ associated with the binary product. Thus, the representing superalgebra is $\bigl(\mathrm{Sym}(R^*), \Delta^{(2)}, \Delta^{(3)} \bigr)$, where we set
\[\Delta^{(2)} := \Delta_\times, \qquad \Delta^{(3)}:= (\Id \otimes S \otimes \Id)\circ (\Delta_+)_2,\]
 where $(\Delta_+)_2 = (\Delta_+\otimes \Id)\circ \Delta_+$.

\begin{Example}
Consider the cohomology ring $R := H^*\bigl(S^2 \times S^3\bigr) \cong \Z[x, \theta]/\bigl(x^2, \theta^2\bigr)$, where we assign $\widetilde{x} =0$ and $\widetilde{\theta}=1$. A~general element of this ring is of the form $a + b x + c \theta + d x\theta$, where $a,b,c,d \in \Z$. We then define the dual generators $v, w \in R^*_\mathbf{0}$ dual to~$1$ and~$x$, respectively, and $\xi, \eta \in R^*_\mathbf{1}$ dual to $\theta$ and $x\theta$, respectively. Then the representing superalgebra is $X \cong \Z[v,w,\xi, \eta]$ together with
\begin{alignat*}{3}
& \Delta^{(2)}(v) = v \otimes v, \qquad && \Delta^{(2)}(w) = w \otimes v + v \otimes w,&\\
& \Delta^{(2)}(\xi) = \xi \otimes v + v \otimes \xi, \qquad && \Delta^{(2)}(\eta) = \eta \otimes v + v \otimes \eta + w \otimes \xi + \xi \otimes w,&
\end{alignat*}
and
\[\Delta^{(3)}(\alpha) = \alpha \otimes 1 \otimes 1 - 1 \otimes \alpha \otimes 1 + 1\otimes 1 \otimes \alpha,\]
for any $\alpha \in \{ v,w, \xi, \eta\}$.
\end{Example}

\begin{Example}
Consider the Clifford algebra $R := Cl_{1|1}(\mathbb{K})$. The generators are $x$ and $\theta$, subject to the relations $x^2=1$, $\theta^2 =0$, and the commutation rule $x \theta = - \theta x$. We assign $\widetilde{x} =0$ and $\widetilde{\theta} = 1$. A general element of this ring is of the form $a + b x + c \theta + d x\theta$, where $a,b,c,d \in \mathbb{K}$. We then define the dual generators $v, w \in R^*_\mathbf{0}$ dual to $1$ and $x$, respectively, and $\xi, \eta \in R^*_\mathbf{1}$ dual to $\theta$ and $x\theta$, respectively. Then the representing superalgebra is $X \cong \Z[v,w,\xi, \eta]$ together with
\begin{align*}
& \Delta^{(2)}(v) = v \otimes v + w \otimes w, && \Delta^{(2)}(w) = v \otimes w + w \otimes v,\\
& \Delta^{(2)}(\xi) = v \otimes \xi + w \otimes \eta + \xi \otimes v - \eta \otimes w, && \Delta^{(2)}(\eta) = v \otimes \eta + w \otimes \xi - \xi \otimes w + \eta \otimes v,
\end{align*}
and
\[\Delta^{(3)}(\alpha) = \alpha \otimes 1 \otimes 1 - 1 \otimes \alpha \otimes 1 + 1\otimes 1 \otimes \alpha,\]
for any $\alpha \in \{ v,w, \xi, \eta\}$.
\end{Example}

\subsection{An equivalence of categories}
Via general categorical nonsense, we expect an anti-equivalence of categories between affine supertrusses and their representing superalgebras. With this in mind, we make the following definition.

\begin{Definition}
A \emph{supercotruss} is a triple $\bigl(X, \Delta^{(2)}, \Delta^{(3)}\bigr)$ that satisfy the conditions~\eqref{eqn:Con}. \emph{Homomorphisms of supercotrusses} $\varphi\colon X \rightarrow X'$ are superalgbebra homomorphisms that satisfy
\begin{subequations}\label{eqn:Morph}
\begin{align}\label{eqn:Morpha}
&\Delta^{' (2)} \circ \varphi = (\varphi \otimes \varphi)\circ \Delta^{(2)},\\ \label{eqn:Morphb}
&\Delta^{' (3)} \circ \varphi = (\varphi \otimes \varphi\otimes \varphi)\circ \Delta^{(3)}.
\end{align}
\end{subequations}
\end{Definition}

The resulting category we denote as $\cO(\catname{ASTru}_{\mathbb{K}})$. If counits and/or cozeros exist, then it is further required that
\begin{equation*}
\varepsilon'_{\mathrm{unit}} \circ \varphi = \varepsilon_{\mathrm{unit}}, \qquad \varepsilon'_{\mathrm{zero}} \circ \varphi = \varepsilon_{\mathrm{zero}}.
\end{equation*}

\begin{Remark}
The notation $\cO(\catname{ASTru}_{\mathbb{K}})$ for the category of supercotrusses has been chosen to remind one that cotrusses should be thought of as the coordinate rings on some superspace in the same way we think of Hopf superalgebras as the coordinate rings of affine supergroups.
\end{Remark}

\begin{Theorem}\label{trm:EqCat}
There is an equivalence of categories between $\catname{ASTru}_{\mathbb{K}}^{\mathrm{op}} \simeq\cO(\catname{ASTru}_{\mathbb{K}})$.
\end{Theorem}

\begin{proof}
The contravariant assignment $\mathcal{F}\colon \catname{ASTru}_{\mathbb{K}}^{\text{op}} \rightarrow \mathcal{O}(\catname{ASTru}_{\mathbb{K}})$ defines an equivalence of categories if it is essentially surjective on objects and is fully faithful (see \cite[IV.4]{MacLane:1988}).

\textbf{Essential surjectivity:} Let $\bigl(X, \Delta^{(2)}, \Delta^{(3)}\bigr)$ be an object in the category of supercotrusses $\mathcal{O}(\catname{ASTru}_{\mathbb{K}})$. By defining the functor $T(-) := \Hom_{\catname{SAlg}_{\mathbb{K}}}(X, -)$, the compatibility conditions~\eqref{eqn:Con} established in Theorem~\ref{trm:RepObs} ensure that $T(A)$ is a truss for each $A \in \catname{SAlg}_{\mathbb{K}}$. Because $T$ is representable by definition, we have $\mathcal{F}(T) \cong X$, and so essential surjectivity is established.

\textbf{Full faithfulness:} Let $T$ and $T'$ be two affine supertrusses represented by the supercotrusses~$X$ and~$X'$, respectively. By Yoneda's lemma, there is a natural bijection between natural transformations $\phi_-\colon T(-) \to T'(-)$ and superalgebra homomorphisms $\varphi\colon X' \to X$, which we claim satisfy \eqref{eqn:Morph}. The conditions \eqref{eqn:Morph} require explicit checking. For the binary multiplication, we have
\[ \phi_A(st) = m_A^{(2)} \circ (s \otimes t) \circ \Delta^{(2)} \circ \varphi = m_A^{(2)} \circ (s \circ \varphi \otimes t \circ \varphi) \circ \Delta'^{(2)} = \phi_A(s)\phi_A(t).\]
Similarly, for the ternary heap operation we have
\begin{align*}
\phi_A([t_1, t_2, t_3]) &= m_A^{(3)} \circ (t_1 \otimes t_2 \otimes t_3) \circ \Delta^{(3)} \circ \varphi = m_A^{(3)} \circ (t_1 \circ \varphi \otimes t_2 \circ \varphi \otimes t_3 \circ \varphi) \circ \Delta'^{(3)} \\
&= [\phi_A(t_1), \phi_A(t_2), \phi_A(t_3)].
\end{align*}
Thus, the components $\phi_A$ are truss homomorphisms if and only if $\varphi$ satisfies the structural compatibility conditions \eqref{eqn:Morph}. In conclusion, $\mathcal{F}$ induces a bijection on the respective hom-sets, proving full faithfulness.
\end{proof}

\begin{Remark}
The reader should recall the well-known fundamental anti-equivalence between the category of affine group schemes over a field $\mathbb{K}$ and the category of commutative Hopf algebras over $\mathbb{K}$. Theorem~\ref{trm:EqCat} is the appropriate generalisation of this to affine supertrusses.
\end{Remark}

\subsection{Affine two-sided superbraces}
We proceed to generalise Brzezi\'nski's \cite[Corollary 3.10]{Brzezinski:2020} construction of two-sided braces from braceable trusses. We assume that the representing supercotruss $X$ of an affine supertruss $T$ comes with a counit $\varepsilon_{\mathrm{unit}}\colon X \rightarrow \mathbb{K}$. Thus, for any superalgebra $A \in \catname{SAlg}_{\mathbb{K}}$, the truss $T(A)$ is a~unital truss with $e := !_A \circ \varepsilon_{\mathrm{unit}}$. Then directly from \cite[Lemma~3.9]{Brzezinski:2020}, we have the following.

\begin{Lemma}\label{lem:Brace}
Let $T$ be an affine unital supertruss, then for every $A \in \catname{SAlg}_{\mathbb{K}}$, the operation $+_e$ together with the binary product satisfies a brace-type distributivity rule, i.e.,
\[s(t+_e u) = st -_e s +_e su, \qquad (t+_e u)s = ts -_e s +_e us\]
for all $s,t,u \in T(A)$.
\end{Lemma}

\begin{Definition}\label{def:Bracable}
A \emph{braceable affine supertruss} is an affine unital supertruss $T\colon \catname{SAlg}_\mathbb{K} \rightarrow \catname{Truss}_1$ such that $V \circ T\colon \catname{SAlg}_\mathbb{K} \rightarrow \catname{Grp}$. That is, the composition of functors is an affine supergroup. The category of braceable affine supertrusses we denote as $\catname{ASTru}_{\mathbb{K}}^{\catname{2SBrac}}$.
\end{Definition}

By construction, $(T(A), +_e, \cdot_A)$ is a two-sided brace for all $A \in \catname{SAlg}_{\mathbb{K}}$, if $T \in \catname{ASTru}_{\mathbb{K}}^{\catname{2SBrac}}$.

\begin{Proposition}\label{prop:BraceNat}
The construction of a two-sided brace from a braceable affine supertruss is functorial.
\end{Proposition}

\begin{proof}
Lemma \ref{lem:Brace} establishes the construction. Let $\psi\colon A \rightarrow B $ be a superalgebra homomorphism. Then, $\psi^T := T(\psi)$ is a truss homomorphism $\psi^T\colon T(A) \rightarrow T(B)$ by the definition of an affine supertruss. Then we observe that
\begin{align*}
& \psi^T(t +_e u) = \psi^T [t,e,u] = [\psi^T(t),\psi^T(e),\psi^T(u)] = \psi^T(t) +_{\psi^T(e)} \psi^T(u),\\
& \psi^T(-_e t) = \psi^T[e,t,e] = [\psi^T(e),\psi^T(t),\psi^T(e)] = -_{\psi^T(e)} \psi^T(t).
\end{align*}
As $\psi^T(e)$ is the unit element in $T(B)$, we observe that $\psi$ induces a homomorphism of two-sided braces, and hence the construction is functorial.
\end{proof}

Given a braceable affine supertruss, we refer to the induced two-sided brace as an \emph{affine two-sided superbrace}, understood as a composition of functors (see \eqref{eqn:BrFunc})
\[ Br \circ T\colon \ \catname{SAlg}_{\mathbb{K}} \rightarrow \catname{2SBrac},\]
where the target category is the category of two-sided braces. Note that any natural transformation $\phi_-\colon T(-) \rightarrow T'(-)$ has components that are unital truss homomorphisms, and so defines a~homomorphism of two-sided braces.

\subsection{Super-Adler's map}
Looking towards solutions of the Yang--Baxter equation from affine two-sided superbraces, we proceed as follows. A solution to the Yang--Baxter equation we understand as a natural transformation
\begin{equation*}
r_- \colon \ T(-)\times T(-) \longrightarrow T(-)\times T(-),
\end{equation*}
such that any component $r_A(s,t) := (\lambda_s(t), \rho_t(s))$ ($A \in \catname{SAlg}_{\mathbb{K}}$) is a solution to the set-theoretic Yang--Baxter equation \eqref{eqn:SetYB} (also see \eqref{eqn:YB1}, \eqref{eqn:YB2} and \eqref{eqn:YB3}). As we are dealing with natural transformations, we do not have a single solution to a single Yang--Baxter equation; rather, we have families parametrised by superalgebras. More correctly, we pass from the set-theoretic Yang--Baxter equation to the Yang--Baxter equation in the category of affine superschemes. We will refer to $r_-$ as a \emph{Yang--Baxter natural transformation}, and any component $r_A$ as a \emph{Yang--Baxter map}.

\begin{Remark}
The first generalisation of the Yang--Baxter equation to algebraic varieties was given by Etingof \cite{Etingof:2003}. Grahovski et al. \cite{Grahovski:2016} constructed Yang--Baxter maps between Grassmann extensions of algebraic varieties.
\end{Remark}

Via Yoneda's lemma, Yang--Baxter natural transformations correspond to superalgebra maps $r^*\colon X \otimes X \rightarrow X \otimes X$ that satisfy the Yang--Baxter equation (in braid form)
\[(r^* \otimes \Id_X)(\Id_X \otimes r^*)(r^* \otimes \Id_X) = (\Id_X \otimes r^*)(r^* \otimes \Id_X)(\Id_X \otimes r^*).\]
We refer to such superalgebra maps as {\emph{co-Yang--Baxter maps}} on the coordinate ring of $T$.

We now proceed to illustrate the concepts by constructing super-versions of Adler's map (see~\cite{Adler:1993}) within the framework of affine supertrusses.

The affine supercotruss we work with is $X = \mathbb{K}[x, x^{-1}, \theta] $, where $\widetilde{x} = 0$ and $\widetilde{\theta} =1$, and, of course, $\theta^2 =0$. The cotruss structure is given by
\begin{align*}
& \Delta^{(2)}(x) = x \otimes x + \theta \otimes \theta, \\
& \Delta^{(2)}(\theta) = x \otimes \theta + \theta \otimes x , \\
& \Delta^{(3)}(x) = x \otimes 1 \otimes 1 - 1 \otimes x \otimes 1 + 1 \otimes 1 \otimes x, \\
& \Delta^{(3)}(\theta) = \theta \otimes 1 \otimes 1 - 1 \otimes \theta \otimes 1 + 1 \otimes 1 \otimes \theta.
\end{align*}
One can deduce that
\[\Delta^{(2)}\bigl(x^{-1}\bigr) = x^{-1}\otimes x^{-1} - x^{-2}\theta \otimes \theta x^{-2}.\]
Note that $s \in T(A)$ can be completely described by its action on the generators; thus
we write $s = (\mathsf{x}, \vartheta)$, with $\mathsf{x} := s(x)$ and $\vartheta := s(\theta)$. With these definitions, the truss structure is given by
\begin{align*}
 st = (\mathsf{x}_s \mathsf{x}_t + \vartheta_s \vartheta_t , \mathsf{x}_s \vartheta_t + \vartheta_s \mathsf{x}_t), \qquad
 [s,t,u] = (\mathsf{x}_s - \mathsf{x}_t + \mathsf{x}_u , \vartheta_s - \vartheta_t + \vartheta_u).
\end{align*}
The counit is given by $\varepsilon_{\textrm{unit}}(x) = 1$ and $\varepsilon_{\textrm{unit}}(\theta) = 0$. We observe that $e = (1_A , 0_A)$.

The multiplicative inverse is given by
\begin{equation*}
s^{-1} = \bigl(\mathsf{x}^{-1}_s, - \mathsf{x}^{-2}_s \vartheta_s\bigr).
\end{equation*}
 As every element $s \in T(A) \cong h^X(A)$ is invertible, $T$ is a braceable affine supertruss (see Definition~\ref{def:Bracable}). Note that the binary multiplication in this example defines the multiplicative affine supergroup. That is, $\cO\bigl(\mathbb{G}^{1|1}_m\bigr) = \mathbb{K}\bigl[x, x^{-1}, \theta\bigr]$. The associated affine two-sided superbrace is then defined using
\begin{equation*}
s+_e t := (\mathsf{x}_s + \mathsf{x}_t -1_A , \vartheta_s+ \vartheta_t), \qquad s-_e t := (\mathsf{x}_s - \mathsf{x}_t + 1_A, \vartheta_s - \vartheta_t).
\end{equation*}
To define a parameter in a natural way, we define $\varepsilon_{\textrm{const}}\colon X \rightarrow \mathbb{K}$ as $\varepsilon_{\textrm{const}}(x) = \gamma$ and ${\varepsilon_{\textrm{const}}(\theta)=0}$. We then set $c = (\gamma 1_A, 0_A) =: (\gamma, 0_A)$, for an arbitrary $A\in \catname{SAlg}_\mathbb{K}$. We then employ the standard form of Adler's map \cite{Adler:1993}
\[\lambda_s(t) := t -_e c (s+_e t)^{-1}, \qquad \rho_t(s) := s +_e c (s+_e t)^{-1}.\]
From the classical theory, we know that $r_A = (\lambda_s(t), \rho_t(s) )$ is a Yang--Baxter map, i.e., we have a solution to the braid form of the set-theoretic Yang--Baxter equation~\eqref{eqn:SetYB}. Moreover, each~$r_A$ is invertible and moreover, involutive, implying that $r_- \colon T(-) \times T(-) \rightarrow T(-)\times T(-)$ is a~natural isomorphism in $\catname{Set}$. Then, via Yoneda's lemma, we have an associated isomorphism $r\colon \mathbb{G}^{1|1}_m \times \mathbb{G}^{1|1}_m \rightarrow \mathbb{G}^{1|1}_m \times \mathbb{G}^{1|1}_m$ of affine superschemes (see \cite[Chapter~10]{Carmeli:2011}).

Writing the Yang--Baxter map out, we have
\begin{subequations}
\begin{align} \label{eqn:SAdlerLam}
& \lambda_s(t) = \bigl(\mathsf{x}_t + 1_A - \gamma (\mathsf{x}_s + \mathsf{x}_t - 1_A)^{-1}, \vartheta_t + \gamma (\mathsf{x}_s + \mathsf{x}_t-1_A)^{-2}(\vartheta_s + \vartheta_t) \bigr),\\ \label{eqn:SAdlerRho}
& \rho_t(s) = \bigl(\mathsf{x}_s - 1_A + \gamma (\mathsf{x}_s + \mathsf{x}_t - 1_A)^{-1}, \vartheta_s - \gamma (\mathsf{x}_s +\mathsf{x}_t-1_A)^{-2}(\vartheta_s + \vartheta_t) \bigr).
\end{align}
\end{subequations}
Directly from the classical theory of Adler's map, we have the linear invariant
\begin{subequations}
\begin{equation}\label{eqn:LinearInv}
I_1(s,t) := \lambda_s(t) +_e \rho_t(s) = s +_e t = (\mathsf{x}_s + \mathsf{x}_t - 1_A, \vartheta_s + \vartheta_t) =: \bigl(I^{\mathbf{0}}_1(s,t) , I^{\mathbf{1}}_1(s,t)\bigr),
\end{equation}
which is unchanged under $(s,t) \mapsto (\lambda_s(t), \rho_t(s))$. Similarly, Adler's rational invariant also extends to the current setting as
\begin{equation}\label{eqn:RationalInv}
I_2(s,t) := \lambda_s(t)\rho_t(s) +_e  (t \lambda_s(t) +_e s \rho_t(s)  )  ( \lambda_s(t) +_e \rho_t(s) )^{-1}.
\end{equation}
\end{subequations}

For clarity, we set $W := I^{\mathbf{0}}_1(s,t)$ and $\Theta := I^{\mathbf{1}}_1(s,t)$, and define
\begin{alignat*}{3}
&\bar{\mathsf{x}}_t := \mathsf{x}_t + 1_A - \gamma W^{-1}, \qquad && \bar{\mathsf{x}}_s := \mathsf{x}_s - 1_A + \gamma W^{-1},&\\
&\bar{\vartheta}_t := \vartheta_t + \gamma \Theta W^{-2},\qquad  && \bar{\vartheta}_s := \vartheta_s -\gamma \Theta W^{-2}.&
\end{alignat*}
Expanding out \eqref{eqn:RationalInv} into its components we obtain
\begin{align*}
&I_2^{\mathbf{0}}(s,t)= \bar{\mathsf{x}}_t \bar{\mathsf{x}}_s + (\mathsf{x}_t \bar{\mathsf{x}}_t + \mathsf{x}_s \bar{\mathsf{x}}_s)W^{-1} - 1_A + \bigl(1_A - 2 (\bar{\mathsf{x}}_t - \bar{\mathsf{x}}_s -1_A)W^{-2} \bigr)\bar{\vartheta}_t \bar{\vartheta}_s,\\
&I_2^{\mathbf{1}}(s,t)= \bar{\mathsf{x}}_t \bar{\vartheta}_s + \bar{\mathsf{x}}_s \bar{\vartheta}_t + W^{-1} \bigl( 2\bar{\mathsf{x}}_t \bar{\vartheta}_t + 2\bar{\mathsf{x}}_s \bar{\vartheta}_s - \bar{\vartheta}_t + \bar{\vartheta}_s + \gamma W^{-1}(\bar{\vartheta}_t - \bar{\vartheta}_s) \bigr) \\ \nonumber
&\hphantom{I_2^{\mathbf{1}}(s,t)=}{} - W^{-2} \bigl( \bar{\mathsf{x}}_t^2 + \bar{\mathsf{x}}_s^2 - \bar{\mathsf{x}}_t + \bar{\mathsf{x}}_s - 1_A + 2\gamma W^{-1}(\bar{\mathsf{x}}_t - \bar{\mathsf{x}}_s) \bigr) \Theta.
\end{align*}
The underlying classical map is defined by setting $A = \mathbb{K}$, we obtain
\[\lambda_x(y) =: u = y +1 - \gamma (x+y -1)^{-1}, \qquad \rho_y(x)=: v = x -1 + \gamma (x+y-1)^{-1},\]
setting $\mathsf{x}_s = x$ and $\mathsf{x}_t = y$ in this case. Shifting the variables and setting $x' = x-1$, $y' = y$, $u' = u -1$ and $v'=v$, recovers the standard form of Adler's map, i.e.,
\begin{equation*}
u' = y' - \gamma (x' +y')^{-1}, \qquad v' = x' + \gamma (x' +y')^{-1}.
\end{equation*}
Clearly, we recover the classical linear and rational invariants for the standard Adler's map $I_1(u', v')$ and $I_2(u',v')$.

\begin{Remark}
The super-Adler's map presented above is distinct from the Grassmann extension of Adler's map of Konstantinou-Rizos and Mikhailov~\cite{Konstantinou-Rizos:2016}. In particular, the motivation and methodology are different. Konstantinou-Rizos and Mikhailov construct their map analytically using a Lax matrix refactorisation problem with Grassmann-valued entries; the motivation was the study of noncommutative discrete potential KdV systems. The motivation here is algebraic and driven by the structure of affine supertrusses and superbraces.
\end{Remark}

Note that as the parameter $c$ is purely even, there is no coupling in the even sectors of \eqref{eqn:SAdlerLam} and \eqref{eqn:SAdlerRho} between even and odd variables. We amend this by modifying the Yang--Baxter maps with a Drinfel'd twist (see Ferri \cite{Ferri:2026} and references therein). We define the superalgebra map $\omega\colon X \rightarrow X \otimes X$ as
\[\omega(x) := 1 \otimes 1 + \gamma \theta \otimes \theta, \qquad \omega(\theta) :=0.\]
Then, for any $A \in \catname{SAlg}_\mathbb{K}$, $\omega$ corresponds to
$\Omega(s,t) = (1_A + \gamma \vartheta_s \vartheta_t, 0_A)$.
Directly, it can be shown that we have a $2$-cocycle, i.e., $\Omega$ satisfies
\[\Omega(s,t) +_e \Omega(s+_e t,u) = \Omega(t,u)+_e \Omega(s, t+_e u),\]
we omit the calculation as it is not illuminating. The twisted Yang--Baxter map
\[r'_A := (\lambda'_s(t), \rho'_t(s))\]
is defined as
\[\lambda'_s(t):= \lambda_s(t) +_e \Omega(s,t), \qquad \rho'_t(s):= \rho_t(s) -_e \Omega(s,t),\]
explicitly we have
\begin{align*}
& \lambda'_s(t) = \bigl(\mathsf{x}_t + 1_A - \gamma (\mathsf{x}_s + \mathsf{x}_t - 1_A)^{-1} + \gamma \vartheta_s \vartheta_t, \vartheta_t + \gamma (\mathsf{x}_s + \mathsf{x}_t-1_A)^{-2}(\vartheta_s + \vartheta_t) \bigr),\\
& \rho'_t(s) = \bigl(\mathsf{x}_s - 1_A + \gamma (\mathsf{x}_s + \mathsf{x}_t - 1_A)^{-1} - \gamma \vartheta_s \vartheta_t, \vartheta_s - \gamma (\mathsf{x}_s +\mathsf{x}_t-1_A)^{-2}(\vartheta_s + \vartheta_t) \bigr).
\end{align*}
Note that the twist does not change the classical underlying map. The twisted Yang--Baxter maps are invertible, but no longer involutive. Moreover, we have the linear invariant
\[I'_1(s,t) := \lambda'_s(t) +_e \rho'_t(s) = s +_e t = (\mathsf{x}_s + \mathsf{x}_t - 1_A, \vartheta_s + \vartheta_t) = I_1(s,t),\]
which is unchanged under $(s,t) \mapsto (\lambda'_s(t), \rho'_t(s))$. Finally, we note that the rational invariant requires modification due to the twisting
\begin{gather*}
I'^{\mathbf{0}}_2(s,t) = I^{\mathbf{0}}_2(s,t) - \gamma \bigl((\mathsf{x}_s- \mathsf{x}_t)\bigl(1_A - W^{-1}\bigr) - 2 \bigl(1_A - \gamma W^{-1}\bigr) \bigr)\vartheta_s \vartheta_t,\!\! \qquad I'^{ \mathbf{1}}_2(s,t) = I^{\mathbf{1}}_2(s,t).
\end{gather*}
The existence of linear and rational invariants (first integrals) in both the untwisted and twisted versions of super-Adler's map strongly suggests that these maps may find applications in the theory of discrete integrable systems with both bosonic and fermionic degrees of freedom.

To sketch a model, we define $s_{n,m} = (\mathsf{x}_{n,m}, \vartheta_{n,m}) \in T(A)$ at each lattice site $(n,m) \in \Z \times \Z$ for some arbitrary $A \in \catname{SAlg}_{\mathbb{K}}$. The local super-Adler transitions are defined as
\[
r'_A(s_{n,m}, s_{n+1, m}) = (s_{n,m+1}, s_{n+1, m+1}),\]
and generates a coupled system of discrete quad-difference equations across the lattice. We have established that the transitions satisfy the Yang--Baxter equation in its braid form, and that we have local invariants $I'_1$ and $I'_2$. While we defer details, such a system may be referred to as the \emph{super discrete potential Korteweg--de Vries system}.\footnote{Here by `super' we just mean that the system contains both bosonic and fermionic variables and we do not directly imply the system is supersymmetric.} Note that the underlying classical system is precisely the dpKdV system (see Adler \cite{Adler:1993} and Adler, Bobenko and Suris \cite{Adler:2003} for further details).

\section{Concluding remarks}
In this paper, we generalised trusses to supermathematics by defining affine supertrusses and their representing superalgebras, which we refer to as supercotrusses; the latter are essentially quantum heaps with a compatible comultiplication. As `sanity checks', we have shown that abelian supergroups and finite projective superrings (over $\mathbb{K}$) provide examples of affine supertrusses. Thus, the category is far from empty. Via braceable affine supertrusses, we constructed affine two-sided superbraces as a super-version of Rump's two-sided braces. As a~potentially physically relevant example, a super-Adler's map has been presented via an affine supertruss/superbrace and a Drinfel'd twist.

\textbf{Further directions:}\vspace{-2mm}
\begin{enumerate}\itemsep=0pt
\item \emph{Paragons:} Brzezi\'nski showed that the correct idea for quotients in the category of trusses is a paragon. Loosely, a paragon $P$ is a sub-heap of a truss $T$ that is closed under an appropriate multiplicative semigroup action which ensures that $T/P$ is itself a truss (see \cite[Proposition 3.22]{Brzezinski:2020}). Extending this construction to affine supertrusses and their `superparagons' requires careful analysis.
For a given $A \in \catname{SAlg}_\mathbb{K}$, the set-theoretical quotient $T(A)/P(A)$ will be a truss. However, quotients in affine supergeometry via functors are notoriously difficult; it is not automatically the case that the quotient of representable functors is representable. It is known that if $G$ is an affine algebraic supergroup and $H$ is a closed normal subsupergroup, the quotient $G/H$ is again an affine supergroup scheme. As paragons are analogous to normal subgroups, we conjecture that the quotient of an affine supertruss by one of its superparagons (assuming they exist) is representable.
\item \emph{Truss modules:} Brzezi\'nski defined truss modules as a heap-theoretic generalisation of a~ring module (see \cite[Definition 4.1]{Brzezinski:2020}). Essentially, we have a representation of a~truss acting on an abelian heap. In the group setting, truss modules are akin to group representations. The notion in the affine supertruss setting is intuitively clear; we have a~representable functor $M(-)\colon \catname{SAlg}_{\mathbb{K}} \rightarrow \catname{AbHeap}$, together with a natural transformation $\triangleright_-\colon T(-) \times M(-) \rightarrow M(-)$, whose components satisfy the axioms of a truss module for each $A \in \catname{SAlg}_{\mathbb{K}}$. We~stress that analysing rings as trusses allows for a richer category of modules. We speculate that, as cohomology rings are superrings, they can be considered as affine supertrusses; their associated truss modules may be of interest in algebraic topology and homological algebra. Details of representability and the representing objects await exploration.
\item \emph{Yang--Baxter supergeometry and discrete models:} The construction of an affine two-sided superbrace from a braceable affine supertruss opens up the notion of Yang--Baxter maps in the setting of affine superschemes, a currently largely unexplored arena. The notion of Yang--Baxter maps here can be formulated in terms of natural transformations. It is desirable to construct further and more involved examples of Yang--Baxter natural transformations. In particular, it would be expedient to understand the Grassmann extensions of Yang--Baxter maps of Adamopoulou and Papamikos \cite{Adamopoulou:2025}, Grahovski, Konstantinou-Rizos and Mikhailov \cite{Grahovski:2016}, Konstantinou-Rizos \cite{Konstantinou-Rizos:2020,Konstantinou-Rizos:2026}, and Konstantinou-Rizos and Mikhailov~\cite{Konstantinou-Rizos:2016}, etc., within the framework suggested here; if that is generally possible. Nonetheless, we are confident that applications of affine supertrusses in integrable systems on lattices with bosonic and fermionic degrees of freedom may be found.
\end{enumerate}

\subsection*{Acknowledgements}
The author thanks Tomasz Brzezi\'nski for introducing him to ternary algebraic structures. Thanks are extended to Bernard Rybołowicz for reading earlier drafts of this work. The author gratefully acknowledges the input of the two anonymous referees, who have helped to significantly improve this paper.
Gratitude is extended to the organisers and audience of the ``Integrability and Conformal Field Theory 2026'' meeting at Cardiff University for the opportunity to present preliminary aspects of this research.

\pdfbookmark[1]{References}{ref}
\LastPageEnding

\end{document}